\documentclass[nolineno]{ieeeaccess}
\usepackage{cite}
\usepackage{amsmath,amssymb,amsthm,mathtools,booktabs,array,graphicx,textcomp}
\usepackage{microtype}
\usepackage{balance}
\usepackage[hidelinks,bookmarks=false]{hyperref}
\hypersetup{pdftitle={The Cost of Missed States: Online Tracking with Current-Observation Resets},pdfauthor={Hiroyasu Usami}}
\newtheorem{theorem}{Theorem}
\newtheorem{proposition}[theorem]{Proposition}
\newtheorem{corollary}[theorem]{Corollary}
\newtheorem{lemma}[theorem]{Lemma}
\theoremstyle{definition}\newtheorem{definition}[theorem]{Definition}
\theoremstyle{remark}
\newcommand{\TV}{\operatorname{TV}}
\newcommand{\osc}{\operatorname{osc}}
\newcommand{\softmax}{\operatorname{softmax}}
\newcommand{\R}{\mathbb R}
\usepackage{etoolbox}
\makeatletter
\patchcmd{\@maketitle}{Digital Object Identifier\space\@doi}{}{}{
  \PackageError{manuscript}{Could not suppress unassigned DOI}{Check official class}}
\patchcmd{\ps@titlepage}{VOLUME\ \thevol, \theyear}{AUTHOR'S MANUSCRIPT}{}{
  \PackageError{manuscript}{Could not set first-page manuscript footer}{Check official class}}
\patchcmd{\ps@titlepage}{VOLUME\ \thevol, \theyear}{AUTHOR'S MANUSCRIPT}{}{
  \PackageError{manuscript}{Could not set verso manuscript footer}{Check official class}}
\patchcmd{\ps@headings}{VOLUME\ \thevol, \theyear}{AUTHOR'S MANUSCRIPT}{}{
  \PackageError{manuscript}{Could not set recto manuscript footer}{Check official class}}
\patchcmd{\ps@headings}{VOLUME\ \thevol, \theyear}{AUTHOR'S MANUSCRIPT}{}{
  \PackageError{manuscript}{Could not set regular verso footer}{Check official class}}
\let\manuscript@titlepagestyle\ps@titlepage
\let\manuscript@headingstyle\ps@headings
\newcommand{\manuscript@oddfoot}[1]{\raisebox{#1}{\hbox to\textwidth{%
  {\footervolfont AUTHOR'S MANUSCRIPT}\hfill{\footerpagefont\thepage}}}}
\newcommand{\manuscript@evenfoot}[1]{\raisebox{#1}{\hbox to\textwidth{%
  {\footerpagefont\thepage}\hfill{\footervolfont AUTHOR'S MANUSCRIPT}}}}
\def\ps@titlepage{\manuscript@titlepagestyle
  \def\@oddfoot{\manuscript@oddfoot{3pt}}\def\@evenfoot{\manuscript@evenfoot{3pt}}}
\def\ps@headings{\manuscript@headingstyle
  \def\@oddfoot{\manuscript@oddfoot{8pt}}\def\@evenfoot{\manuscript@evenfoot{8pt}}}
\expandafter\patchcmd\csname\string\IEEEbiography\endcsname{\vskip \@IEEEBIOskipN plus 1fil minus 0\baselineskip}{\vskip 1.5\baselineskip}{}{
  \PackageError{manuscript}{Could not remove elastic biography gap}{Check official class}}
\makeatother
\AddToHookNext{shipout/foreground}{%
  \begingroup
  \setlength{\unitlength}{1bp}%
  \put(36.17,-189){\makebox(0,0)[lt]{%
    \parbox[t]{\textwidth}{\normalfont\rmfamily\fontsize{8}{9.6}\selectfont
      \parindent=0pt\parskip=0pt\color{black}%
      This work has been submitted to the IEEE for possible publication.\par
      Copyright may be transferred without notice, after which this version may no longer be accessible.}}}%
  \endgroup
}
\begin{document}
\history{}
\doi{}
\title{The Cost of Missed States:\\Online Tracking with Current-Observation Resets}
\author{\uppercase{Hiroyasu Usami}}
\address{Department of Computer Science, College of Engineering, Chubu University, Kasugai, Aichi 487--8501, Japan}
\corresp{Corresponding author: Hiroyasu Usami (e-mail: usami@fsc.chubu.ac.jp).}
\markboth{Usami: The Cost of Missed States}{Usami: The Cost of Missed States}
\begin{abstract}
An online tracker retains one state and may reset only to the current observation, although its complete observed history is available. Offline knows the future and obeys the same reset rule. Under a hard discrepancy radius, for every horizon $T\ge5$ and every feasible randomized causal policy, a fixed path on a known five-point metric requires at least $(T+1)/2$ expected refreshes, against exactly two offline refreshes with a strict radius margin. Missing the useful prefix center forces a refresh at every suffix observation. The path visits four values, with resolution fixed as $T$ grows. Write $J_\lambda$ for cumulative discrepancy plus a charge $\lambda>0$ per refresh, and $V_T$ for the sum of successive observation distances. A candidate prefix adapts Bienkowski and Schmid's randomized lower bound to this objective under current resets, even with a known horizon and $V_T<3\lambda$. This construction needs finer grids; a fixed finite metric with positive minimum separation instead admits a competitive policy for penalized tracking. An anytime policy satisfies $J_\lambda\le\lambda+4\sqrt{\lambda TV_T}$ without knowing the horizon or final variation. Clipping for a known horizon gives $J_\lambda\le\lambda T$ and a factor-three comparison with the minimum of three analytical bounds, without knowing final variation. Current-value diameter and readout sensitivity give same-input output bounds for mixtures, including squared-distance softmax.
\end{abstract}
\titlepgskip=-21pt
\maketitle

\section{Introduction}
An algorithm may remember a past observation without being able to select it as its current state. We study tracking with one retained state and two actions: keep that state or replace it with the current observation. The entire observed history is available. Offline has the same actions and knows the future. Here $T$ is the number of observations and $V_T$ the sum of successive observation distances. We minimize either the refresh count under a hard discrepancy radius or $J_\lambda$, the cumulative state discrepancy plus a charge $\lambda$ per refresh. Current-observation resets can destroy offline competitiveness while still allowing variation-based cost control, which transfers to output error through a suitable sensitivity bound.

Classical tracking models allow representatives that need not coincide with the current observation \cite{bienkowskischmid,chenrussell,yizhang}. Building on Chen and Russell's candidate-selection argument, Theorem~\ref{thm:hardimpossible} gives a separation on a fixed five-point metric. Two possible centers appear in a common prefix. Offline selects the center suited to the eventual suffix while it is available; a policy that misses it must refresh repeatedly. For every $T\ge5$ and every feasible randomized causal policy, a fixed path requires at least $(T+1)/2$ expected refreshes, against exactly two offline refreshes with a strict radius margin. The state space and its resolution stay fixed as $T$ grows. On this family, historical resets give a two-refresh causal policy under the same observations.

For cumulative discrepancy plus a refresh charge $\lambda>0$, Theorem~\ref{thm:softimpossible} adapts Bienkowski and Schmid's nested-interval lower bound. We expose every candidate center in a common prefix, allowing offline to select its eventual center at a legal reset time with uniformly bounded cost. The resulting lower bound applies to randomized causal policies with a known horizon and $V_T<3\lambda$. Its grid becomes finer as the requested ratio increases. On a fixed finite metric with positive minimum separation, refreshing whenever the observation changes is competitive for this penalized objective.

Variation still controls the absolute cost. Without knowing $T$ or the final variation, the policy in Theorem~\ref{thm:anytime} satisfies
\begin{equation}\label{eq:display-01}
 J_\lambda\le\lambda+4\sqrt{\lambda TV_T}
\end{equation}
on every metric-valued path. When $T$ is known, Theorem~\ref{thm:unknownv} bounds $J_\lambda$ by $\lambda T$ and gives a factor-three comparison with the minimum of three analytical cost bounds. The ramp calculation in Theorem~\ref{thm:ramp} and Corollary~\ref{cor:minimax} show that the worst-case order is already unavoidable for offline in the stated regime.

Section~\ref{sec:fidelity} develops output upper bounds for retained operators evaluated at the same current input. For mixtures, current-value diameter and readout sensitivity relate state drift to output discrepancy. We derive sharp local coefficients for squared-distance softmax and give an example in which root-mean-square (RMS) drift over the complete summary tends to zero while the output discrepancy stays fixed.

\section{Model and objectives}\label{sec:setting}
For an integer horizon $T\ge1$, the controller observes a summary $z_t$ in a metric space $(\mathcal Z,d)$ and an input $x_t$ at times $t=0,\ldots,T-1$. Write $z=(z_0,\ldots,z_{T-1})$ for the observed state path. A refresh constructs $A(z_t)$ and replaces the retained state before producing the output, where $\mathcal A$ is the operator space and $A:\mathcal Z\to\mathcal A$ is fixed. The readout $F_t(x,a)$ maps the current input $x$ and installed operator $a$ to an output. If $\tau_t$ is the latest refresh time, the fresh and retained outputs are
\begin{equation}\label{eq:outputs}
\begin{aligned}
 y_t^{\rm fr}&=F_t(x_t,A(z_t)),\\
 y_t^{\rm ca}&=F_t(x_t,A(z_{\tau_t})),\qquad \tau_t\le t.
\end{aligned}
\end{equation}
The initial refresh at $0$ is mandatory and charged. At time $t$, a causal policy bases its decision only on observations through $t$ and its own random choices. Offline knows the complete sequence. Both can keep the active operator or call \emph{refresh}, also called a \emph{reset}, which installs $A(z_t)$ and takes no historical argument.

Historical resets add actions without adding information to a controller that already knows the observation history. On the family in Theorem~\ref{thm:hardimpossible}, they let a causal policy install the center identified by the first suffix observation and use exactly two refreshes.

Constructor parameters, input identities and ordering, and normalization are held fixed or included in $z_t$. Each readout uses the current $x_t$.

The lower bounds permit unrestricted storage of the observed history. The threshold policies below need only the retained summary, the previous summary, accumulated variation, and a constant number of scalar time indices, in addition to the active operator.

\begin{definition}[Variation and tracking functional]
Set
\begin{equation}\label{eq:tracking}
\begin{aligned}
 V_T&=\sum_{t=1}^{T-1}d(z_t,z_{t-1}),\\
 N&=|\{t:\tau_t=t\}|,\\
 J_\lambda&=\sum_{t=0}^{T-1}d(z_t,z_{\tau_t})+\lambda N.
\end{aligned}
\end{equation}
Here $\lambda\ge0$ weights refresh cost relative to metric discrepancy.
\end{definition}
The first term in \eqref{eq:tracking} is cumulative state discrepancy; the second charges for refreshes. We write $N_\pi(z)$ and $J_\lambda(\pi,z)$ for the count and cost of policy $\pi$ on path $z$. On a fixed path, expectations average over the policy's random choices. 

For a hard radius $\eta>0$, feasibility requires $d(z_t,z_{\tau_t})\le\eta$ at every opportunity. The objective is then the refresh count $N$, and $N^*_\eta(z)$ denotes the minimum feasible offline count. The penalized objective $J_\lambda$ instead charges discrepancy. Its minimum over offline schedules with the same reset actions is denoted by $J^*_\lambda(z)$. Unless stated otherwise, there is no maximum-age constraint. Competitive analysis compares a policy's expected cost with a fixed multiple of the offline optimum, possibly plus an additive constant; uniform constants do not depend on the horizon or path. An oblivious path is selected without seeing the policy's realized random choices.

The lower bounds concern tracking in the state metric; Section~\ref{sec:fidelity} gives sufficient conditions for bounds on output discrepancy at the same current input, values, and readout. Extending the certificate to reuse-dependent trajectories would additionally require a bound on error propagation through the underlying dynamics.

\section{Competitive limits}\label{sec:offline}
\subsection{A missed center forces repeated refreshes}
\begin{theorem}[No horizon-independent hard-radius competitive ratio]\label{thm:hardimpossible}
Fix $\eta>0$ and an integer $T\ge5$. Equip the known set
\begin{equation}\label{eq:display-02}
 \mathcal X_\eta=\frac\eta6\{-7,-2,0,2,7\}
\end{equation}
with the usual distance. Let $\pi$ be any randomized causal single-cache policy feasible almost surely on every $\mathcal X_\eta$-valued path at radius $\eta$. The initial refresh is mandatory; subsequent resets use only the current observation, and there is no age constraint. The policy may know $T$, $\eta$, and $\mathcal X_\eta$. There exists a fixed length-$T$ path $z$, visiting at most four values and chosen without observing realized policy coins, such that
\begin{equation}\label{eq:hard-gap}
\begin{gathered}
 N^*_{\eta}(z)=2,\qquad \mathbb E[N_\pi(z)]\ge\frac{T+1}{2},\\
 \frac{\mathbb E[N_\pi(z)]}{N^*_{\eta}(z)}\ge\frac{T+1}{4}.
\end{gathered}
\end{equation}
The offline optimum obeys the same reset rule, and a two-refresh schedule has discrepancy at most $5\eta/6$ throughout. Thus no multiplicative competitive bound independent of $T$ exists, even on this fixed five-point metric against an oblivious path.
\end{theorem}
\begin{proof}
Give the common prefix $z_0=0,z_1=a=-\eta/3,z_2=b=\eta/3$. At its end the retained point $P$ belongs to $\{0,a,b\}$. Put $p_0=\Pr(P=0)$ and choose $c\in\{a,b\}$ so that $q=\Pr(P=c)\le(1-p_0)/2$. Append $L=T-3$ points alternating between $y=7c/2$ and $0$, starting with $y$. The full path depends on the policy's law but not its realized coins.

On $P\ne c$, the first suffix point is at distance at least $7\eta/6>\eta$ from $P$, forcing a refresh. Every later suffix step has length $7\eta/6$, so all $L$ suffix observations force refreshes on this event. The prefix uses at least one refresh when $P=0$ and at least two when $P\ne0$. Since $L\ge2$,
\begin{equation}\label{eq:display-03}
\begin{aligned}
 \mathbb E[N_\pi]
 &\ge2-p_0+L(1-q)\\
 &\ge2+L/2+p_0(L/2-1)\\
 &\ge2+L/2=(T+1)/2.
\end{aligned}
\end{equation}
No independence between the prefix cost and $P$ is used.

Offline refreshes at $0$ and at the prefix occurrence of $c$. Before that occurrence its discrepancy is at most $\eta/3$; the remaining prefix discrepancy is at most $2\eta/3$. The suffix discrepancies are $5\eta/6$ at $y$ and $\eta/3$ at $0$. This schedule therefore uses two refreshes at radius $5\eta/6$. Retaining only $z_0=0$ fails at the first suffix point because $|y|=7\eta/6>\eta$, proving $N^*_\eta=2$.
\end{proof}
Each realized path visits four values: in units of $\eta/6$, either $\{-7,-2,0,2\}$ or $\{-2,0,2,7\}$. Its diameter is $3\eta/2$, and its minimum separation is $\eta/3$. The growing count gap requires neither a larger ambient set nor finer values as $T$ grows; the offline radius retains margin $\eta/6$. Figure~\ref{fig:missed} shows the mechanism.

\paragraph{Historical resets on the constructed family}
Suppose a refresh may install any previously observed state. One deterministic causal policy serves both branches of Theorem~\ref{thm:hardimpossible} with exactly two refreshes: retain $0$ through the prefix, then use the first suffix observation $y$ to identify and install the remembered $c=2y/7$ before producing the output. The prefix discrepancy is at most $\eta/3$; thereafter it is $5\eta/6$ at $y$ and $\eta/3$ at $0$. One refresh fails since $|y|=7\eta/6>\eta$. Historical resets therefore remove the growing count gap on this family.

\begin{figure*}[t]
\centering
\includegraphics[width=\textwidth]{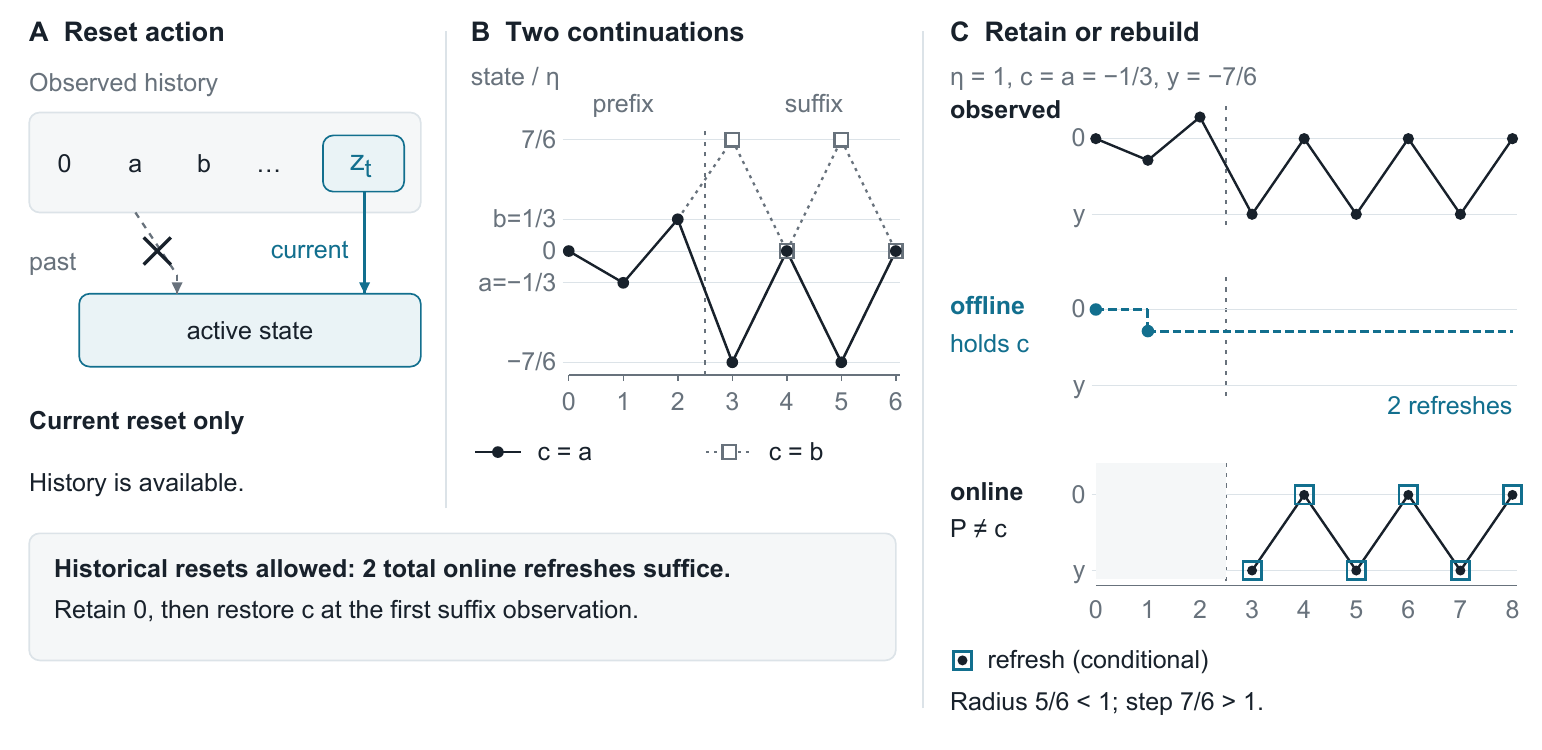}
\caption{The reset constraint in Theorem~\ref{thm:hardimpossible}. (A) History is available; a refresh selects the current observation. (B) A common prefix leads to two suffixes. The branch is chosen from the policy's retained-state distribution, independently of realized policy coins. (C) For $\eta=1$ and $c=a=-1/3$, offline holds the center within radius $5/6$. The online prefix is omitted. Conditional on $P\ne c$, every suffix observation forces a refresh; consecutive endpoints are $7/6>1$ apart. The expected count includes prefix refreshes.}
\label{fig:missed}
\end{figure*}

\subsection{Penalized tracking}
\begin{theorem}[Penalized noncompetitiveness under current resets]\label{thm:softimpossible}
For every $\lambda>0$ and $C\ge1$, there is a finite horizon $T_C$ such that every randomized causal policy, even when given $T_C$ and $\lambda$, has a fixed real-valued path of length $T_C$ with
\begin{equation}\label{eq:penalized-gap}
\begin{gathered}
 z_t\in[0,\lambda],\qquad V_{T_C}<3\lambda,\\
 \mathbb E[J_\lambda(\pi,z)]>C J^*_\lambda(z).
\end{gathered}
\end{equation}
Both policy and offline optimum retain one state, refresh initially, and may subsequently reset only to the current observation. There is no age constraint. The path may depend on the policy's law but is chosen without observing its realized coins. No finite competitive constant, even with an additive allowance uniform in the horizon and path, exists in this class.
\end{theorem}
We prepend the candidate centers to the nested-interval construction of Bienkowski and Schmid \cite{bienkowskischmid}, allowing offline to install its eventual center when observed. In the proof, $M$ is the number of candidate centers, $\delta$ their spacing, and $b$ both the number of phases and the branching factor (Figure~\ref{fig:stages}). Each phase costs online at least $\lambda/16$ in expectation, while offline costs less than $7\lambda/2$. Thus $b>56C$, with $56=16(7/2)$, gives the requested ratio. The choice $M=b^b$ permits $b$ equal splits down to one candidate; $\delta=\lambda/M^2$ keeps the prefix distortion below $\lambda/2$.

\begin{proof}
Fix $\lambda>0$ and $C\ge1$, and set
\begin{equation}\label{eq:display-04}
 b=\lceil56C\rceil+1,\qquad M=b^b,\qquad \delta=\frac{\lambda}{M^2}.
\end{equation}
The common prefix is $\delta(0,1,\ldots,M-1)$. Starting with $R_0=\{0,\ldots,M-1\}$, at level $i=1,\ldots,b$ divide $R_{i-1}$ into $b$ equal contiguous children and choose $R_i$ uniformly, independently of earlier choices conditional on its parent. Put $r_i=\min R_i$. Phase $i$ repeats $\delta r_i$ a fixed number $L_i$ of times. The horizon is
\begin{equation}\label{eq:display-05}
 T_C=M+\sum_{i=1}^b M b^{i-1}
     =M+\frac{M^2-M}{b-1},
\end{equation}
where
\begin{equation}\label{eq:display-06}
 m_i=|R_{i-1}|=\frac{M}{b^{i-1}},\qquad
 L_i=\frac{M^2}{m_i}=M b^{i-1}.
\end{equation}
All lengths are integers and independent of the policy and child choices. Sample the complete nested sequence before execution. This defines an oblivious finite input distribution. All auxiliary observed inputs are identical across its branches. The policy may know this distribution and the horizon, but not the unrevealed choices.

\smallskip\noindent\emph{A legal low-cost comparator.}
Let $x$ be the sole element of $R_b$. Offline refreshes initially and at the prefix occurrence of $\delta x$, omitting the second refresh if $x=0$, and then holds $\delta x$. Its unscaled prefix distortion is
$x(x-1)/2+(M-1-x)(M-x)/2\le M(M-1)/2$.
Because $x\in R_i$, each phase has discrepancy at most $\delta(|R_i|-1)$. Thus
\begin{equation}\label{eq:display-07}
\begin{aligned}
 J^*_\lambda(z)
 &\le2\lambda+\frac{\delta M(M-1)}2
       +\sum_{i=1}^b\delta L_i(|R_i|-1)\\
 &<2\lambda+\frac\lambda2+\lambda
   =\frac{7\lambda}{2}.
\end{aligned}
\end{equation}
Every reset occurs when its state is observed.

\smallskip\noindent\emph{Every phase is costly online.}
Condition on the observations and policy randomness before phase $i$, and write $p$ for the unscaled retained state. The $b$ possible left endpoints are spaced $m_i/b$ apart. Since $b\ge57$, at most $\lfloor b/2\rfloor+1$ endpoints lie in an interval of length $m_i/2$, leaving at least a quarter with $|p-r_i|\ge m_i/4$. If the policy refreshes during this phase it pays at least $\lambda$; otherwise it incurs distortion $\delta L_i|p-r_i|$. Consequently, with $J_i$ denoting phase cost,
\begin{equation}\label{eq:display-08}
\begin{gathered}
 J_i\ge\min\{\lambda,\delta L_i|p-r_i|\},\\
 \Pr\{|p-r_i|\ge m_i/4\mid\text{past}\}\ge\tfrac14,\\
 \mathbb E[J_i\mid\text{past}]\ge\lambda/16.
\end{gathered}
\end{equation}
The factor $1/16$ is the product of the probability $1/4$ and the cost fraction $1/4$, using $\delta L_i m_i=\lambda$. This bound allows the current observation to be seen before deciding and holds for every subsequent random choice of the policy. Summing over phases gives
\begin{equation}\label{eq:display-09}
 \mathbb E_{z,\pi}[J_\lambda(\pi,z)]
 \ge\frac{b\lambda}{16}>\frac{7C\lambda}{2}.
\end{equation}
The distribution has finite support and is independent of realized policy coins. Hence some fixed path has policy-expected cost at least $b\lambda/16$. Together with \eqref{eq:display-07}, this proves \eqref{eq:penalized-gap}.

All states lie in $[0,\delta(M-1)]\subset[0,\lambda]$. The suffix endpoints are nondecreasing. The prefix, its drop to $r_1$, and the remaining suffix each contribute at most $\delta(M-1)$, so every path satisfies
\begin{equation}\label{eq:display-10}
 V_{T_C}\le3\delta(M-1)<\frac{3\lambda}{M}<3\lambda.
\end{equation}
For any additive allowance $A\ge0$, use $b=\lceil56C+16A/\lambda\rceil+1$ in the same construction. A fixed path then satisfies $\mathbb E[J_\lambda(\pi,z)]>C J^*_\lambda(z)+A$; the horizon depends only on $C$ and $A/\lambda$.
\end{proof}
The grid grows with $C$, with spacing
\begin{equation}\label{eq:display-11}
\begin{gathered}
 \delta_{\min}=\frac{\lambda}{M^2},\qquad
 \log_2\frac{\lambda}{\delta_{\min}}=2b\log_2 b\\
 =\Theta(C\log C)\qquad(C\to\infty).
\end{gathered}
\end{equation}
This construction uses a growing grid and increasing precision, whereas Theorem~\ref{thm:hardimpossible} keeps both fixed. In grid units the refresh charge is $C_{\rm grid}=\lambda/\delta_{\min}=M^2=b^{2b}$. The ratio lower bound $b/56$ therefore has order $\log C_{\rm grid}/\log\log C_{\rm grid}$, matching the order of the Bienkowski--Schmid construction \cite{bienkowskischmid}. The prefix preserves that lower bound while making offline's selected center a legal current-observation reset. The stated parameters give $M=57^{57}\approx1.22\times10^{100}$ and $T_C\approx2.64\times10^{198}$ already at $C=1$.

\begin{figure*}[t]
\centering
\includegraphics[width=\textwidth]{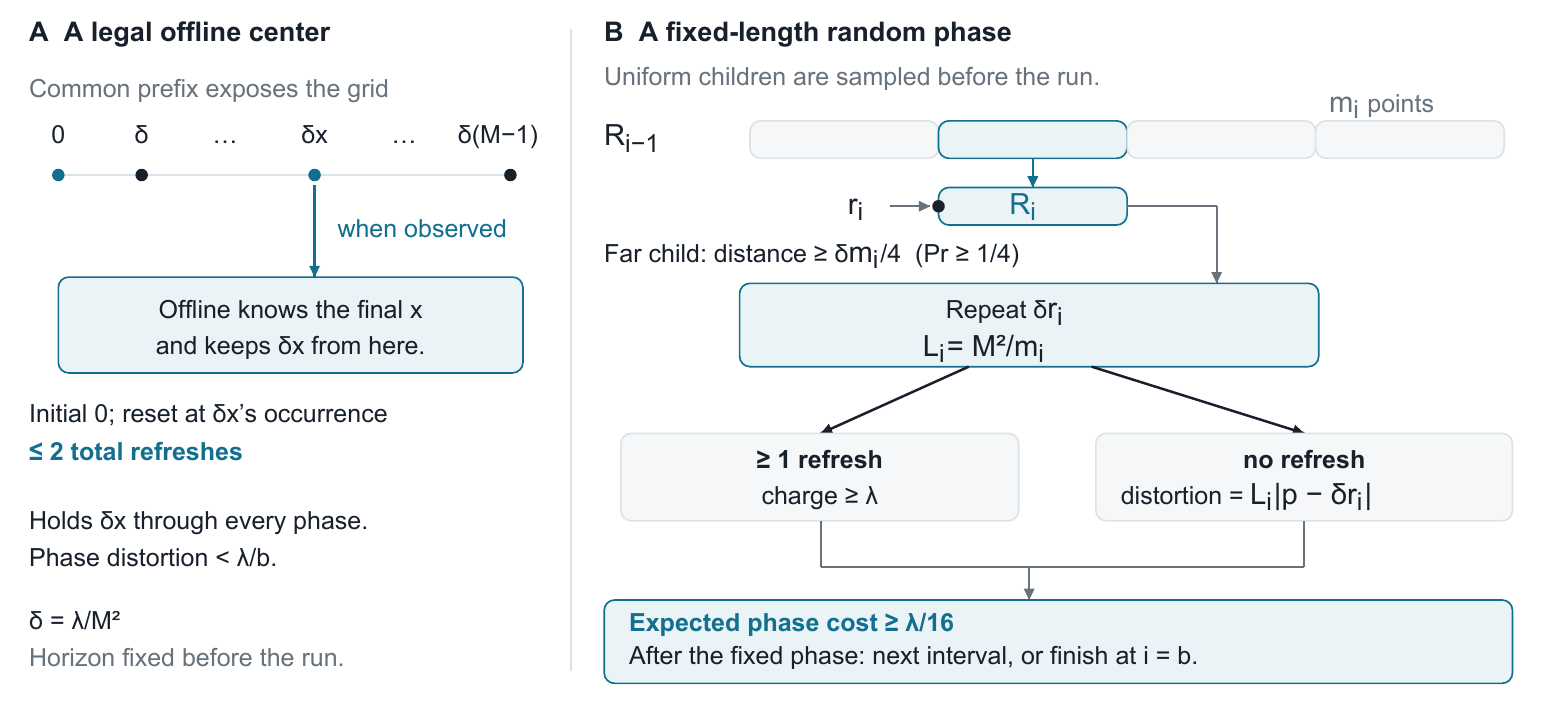}
\caption{A candidate prefix for the nested-interval construction in Theorem~\ref{thm:softimpossible}. The prefix exposes every candidate, so offline can install the final center when observed. Each fixed-length phase repeats the left endpoint of a uniformly chosen smaller child interval. Online incurs a refresh charge or, without a refresh, the full phase distortion. Its conditional expected phase cost is at least $\lambda/16$. Offline holds the final center at total cost below $7\lambda/2$, including prefix and initialization. Child choices, phase lengths, and horizon are fixed before execution.}
\label{fig:stages}
\end{figure*}

\paragraph{Fixed resolution}
On a fixed finite metric with minimum nonzero separation $\delta>0$, refreshing initially and whenever the observation changes gives a $\max\{1,2\lambda/\delta\}$-competitive policy for $J_\lambda$. To see this, let $Q$ count observation changes and let an offline schedule have discrepancy $D$, refresh count $N$, and post-decision errors $e_t$. At a change with no offline refresh, the retained state is unchanged, so $\delta\le e_{t-1}+e_t$. Each error occurs at most twice in this sum, giving $Q\le N-1+2D/\delta$. The policy has zero discrepancy, hence its cost is at most $\lambda N+(2\lambda/\delta)D$, which is bounded by $\max\{1,2\lambda/\delta\}(D+\lambda N)$. Fixed resolution and fixed charge therefore admit a bound uniform in the horizon. These bounds leave the best dependence on resolution undetermined.

\subsection{An exact offline benchmark}
Every single-cache schedule partitions the fully observed path into segments, giving an exact offline benchmark for the examples in Section~\ref{sec:tracking}. Create vertices $0,\ldots,T$. The edge $s\to u$, $s<u$, means refreshing at $s$ and holding through $u-1$, with segment cost $c(s,u)$. The minimum cost through $u-1$ obeys
\begin{equation}\label{eq:offline-oracle}
\begin{gathered}
 c(s,u)=\lambda+\sum_{t=s}^{u-1}d(z_t,z_s),\\
 D_0=0,\qquad D_u=\min_{0\le s<u}\{D_s+c(s,u)\}.
\end{gathered}
\end{equation}
Thus $D_T=J^*_\lambda(z)$ is an exact offline oracle for $J_\lambda$. Updating segment costs incrementally uses $O(T^2)$ distance evaluations and arithmetic operations, with $O(T)$ auxiliary storage beyond the observed path. The time bound treats each distance evaluation as one unit; high-dimensional distance computation and exact-arithmetic bit costs are separate. An integer maximum segment length $K\ge1$ restricts edges to $u-s\le K$.

For a hard output tolerance $\varepsilon$, suppose $B_{s,t}$ is a valid current-input output certificate for retaining state $s$ at $t$. Keep only edges satisfying $\max_{s\le t<u}B_{s,t}\le\varepsilon$, and assign edge cost one to minimize refresh count within this certified class. If the certificates are conservative, this class can be smaller than the set of schedules satisfying the actual output tolerance. In the simple metric-radius version, $B_{s,t}=d(z_t,z_s)$ and the tolerance is $\eta$. Time-dependent reconstruction and controller costs can also be accumulated on edges when their segment dependence is declared. The oracle uses the full path and provides an instance-specific benchmark.

\section{Tracking under bounded variation}\label{sec:tracking}
Theorem~\ref{thm:softimpossible} rules out a constant ratio to offline, while the following variation bounds control absolute cost through $T$ and $V_T$. Table~\ref{tab:scope} summarizes their scope; Figure~\ref{fig:analytic} compares update decisions on two explicit paths.

\begin{table*}[t]
\caption{Main results. Counts include the initial refresh; no age constraint is imposed.}
\label{tab:scope}
\centering\footnotesize
\begin{tabular}{>{\raggedright\arraybackslash}p{0.13\textwidth}>{\raggedright\arraybackslash}p{0.20\textwidth}>{\raggedright\arraybackslash}p{0.22\textwidth}>{\raggedright\arraybackslash}p{0.35\textwidth}}
\toprule
Result & Policy class & State space or precision & Guarantee \\
\midrule
Theorem~\ref{thm:hardimpossible} & Randomized feasible causal & Fixed known five-point metric & For each policy and $T\ge5$, some path has $\mathbb E[N]\ge(T+1)/2$ and $N^*_{\eta}=2$ \\[3pt]
Theorem~\ref{thm:softimpossible} & Randomized causal & Growing candidate grid with finer resolution & Penalized objective: no uniform competitive constant, with a known horizon and $V_T<3\lambda$ \\[3pt]
Theorem~\ref{thm:anytime} & Deterministic causal & Any metric & Anytime cost $J_\lambda\le\lambda+4\sqrt{\lambda TV_T}$; horizon and final variation unknown \\[3pt]
Theorem~\ref{thm:unknownv} & Deterministic causal, known horizon & Any metric & Unknown final variation: $J_\lambda\le\lambda T$ and factor three against the minimum of the analytical bounds in \eqref{eq:clipped-cost} \\[3pt]
Corollary~\ref{cor:minimax} & Randomized causal policies allowed & Real paths with variation budget $V>0$ & Minimax order $\Theta(\sqrt{\lambda TV})$ when $v\le\lambda\le vT^2$, where $v=V/(T-1)$ and $T\ge2$ \\
\bottomrule
\end{tabular}
\end{table*}

The threshold policy refreshes at $0$ and whenever $d(z_t,z_{\tau_{t-1}})>\eta$, before producing the output. Optionally it also refreshes when $t-\tau_{t-1}\ge K$, for an integer $K\ge1$. Age is measured in the opportunity index $t$.

\begin{theorem}[Variation accounting]\label{thm:variation}
For $\eta>0$, the threshold-and-age policy satisfies
\begin{equation}\label{eq:variation-accounting}
\begin{gathered}
 d(z_t,z_{\tau_t})\le\eta,\\
 N\le\min\left\{T,\ 1+\frac{V_T}{\eta}+\left\lfloor\frac{T-1}{K}\right\rfloor\right\}.
\end{gathered}
\end{equation}
Without the age rule, the final term is absent.
\end{theorem}
\begin{proof}
At a refresh the discrepancy is zero; otherwise it is at most $\eta$. Classify each later refresh as a drift event if the strict drift test holds, and otherwise as an age event. Between the preceding refresh and each drift event, path variation exceeds $\eta$ by the triangle inequality. Those intervals are disjoint, so the number of drift events is at most $V_T/\eta$. Each age interval has at least $K$ planned steps. Their disjoint lengths total at most $T-1$, giving the floor term. Add the initial refresh and the trivial $T$ cap.
\end{proof}
This variation argument is closely related to send-on-delta reporting \cite{miskowicz}. With the age cap in force, if $d(z_t,z_s)\le\nu(t-s)$ for some rate $\nu\ge0$ and all $0\le s\le t<T$, the post-decision discrepancy is at most $\min\{\eta,\nu(K-1)\}$. Section~\ref{sec:fidelity} converts this bound to output error through readout sensitivity.

\begin{theorem}[Tracking distortion with a refresh charge]\label{thm:upper}
Without an age constraint,
\begin{equation}\label{eq:fixed-bound}
J_\lambda\le T\eta+\lambda(1+V_T/\eta).
\end{equation}
For known $V_T>0$ and $\lambda>0$, choosing $\eta=\sqrt{\lambda V_T/T}$ gives
\begin{equation}\label{eq:display-12}
 J_\lambda\le\lambda+2\sqrt{\lambda T V_T}.
\end{equation}
Using a known budget $\overline V\ge V_T$ gives the same bound with $\overline V$. For $V_T=0$, the initial-only policy has $J_\lambda=\lambda$. For $\lambda=0$, refreshing always gives $J_0=0$.
\end{theorem}
\begin{proof}
The distortion is at most $T\eta$, and Theorem~\ref{thm:variation} bounds the charge. Minimize the two terms $T\eta+\lambda V_T/\eta$. The edge cases follow directly from the definitions.
\end{proof}
Refreshing at every opportunity costs $\lambda T$; retaining only the initial state costs at most $\lambda+(T-1)V_T$. The optimized fixed threshold requires the final variation or a known upper bound; the anytime rule below uses cumulative observed variation. An age constraint adds at most $\lambda\lfloor(T-1)/K\rfloor$ to the threshold upper bound and can rule out the initial-only schedule.

\subsection{An anytime policy}
In the adaptive rules, $S_t$ denotes variation observed through time $t$, with $S_0=0$.

\begin{theorem}[Anytime variation adaptation]\label{thm:anytime}
Let $\lambda>0$ be known. Refresh at $0$, and for $t\ge1$ define
\begin{equation}\label{eq:anytime-rule}
\begin{gathered}
 S_t=\sum_{j=1}^t d(z_j,z_{j-1}),\\
 \eta_t=\sqrt{\lambda S_t/t},\qquad t\ge1.
\end{gathered}
\end{equation}
After observing $z_t$, refresh exactly when $d(z_t,z_{\tau_{t-1}})>\eta_t$. With no age constraint, this causal policy, without knowing the horizon or final variation, satisfies at every finite horizon $T\ge1$
\begin{equation}\label{eq:anytime-cost}
J_\lambda\le\lambda+4\sqrt{\lambda T V_T}.
\end{equation}
\end{theorem}
\begin{proof}
We bound distortion by the threshold and charge each refresh to newly accumulated variation. A square-root telescoping argument then bounds the refresh count.
Write $V=V_T$. If $T=1$ or $V=0$, the only refresh is the initial one and $J_\lambda=\lambda$. Otherwise, post-decision discrepancy is at most $\eta_t$, so total distortion is at most
\begin{equation}\label{eq:display-13}
 \sqrt{\lambda V}\sum_{t=1}^{T-1}t^{-1/2}\le2\sqrt{\lambda VT}.
\end{equation}
Let $r_0=0<r_1<\cdots<r_m$ be refresh times. For every $j\ge1$,
\begin{equation}\label{eq:display-14}
 \Delta_j:=S_{r_j}-S_{r_{j-1}}
 \ge d(z_{r_j},z_{r_{j-1}})
 >\sqrt{\lambda S_{r_j}/r_j}
 \ge\sqrt{\lambda S_{r_j}/T}.
\end{equation}
Each such endpoint has $S_{r_j}>0$. Thus
\begin{equation}\label{eq:display-15}
 1<\sqrt{T/\lambda}\,\frac{\Delta_j}{\sqrt{S_{r_j}}}
 \le2\sqrt{T/\lambda}\left(\sqrt{S_{r_j}}-\sqrt{S_{r_{j-1}}}\right).
\end{equation}
If $m\ge1$, sum and telescope to get $m<2\sqrt{TV/\lambda}$. If $m=0$, the weak bound holds directly. Consequently the refresh charge is at most $\lambda+2\sqrt{\lambda TV}$, proving the result.
\end{proof}
Distortion and the refresh charge $\lambda(N-1)$ beyond initialization are each at most $2\sqrt{\lambda TV_T}$.

\subsection{Clipped thresholds for a known horizon}
For a known horizon, clipping handles the small- and large-variation regimes using cumulative observed variation, with no age constraint.

\begin{theorem}[Known-horizon adaptation to unknown variation]\label{thm:unknownv}
Let $T\ge2$ and $\lambda>0$ be known. For any metric path define
\begin{equation}\label{eq:clipped-rule}
\begin{gathered}
 S_t=\sum_{j=1}^{t}d(z_j,z_{j-1}),\\
 \eta(v)=\min\{\lambda,\max\{\lambda/T,\sqrt{\lambda v/T}\}\}.
\end{gathered}
\end{equation}
Refresh initially and then exactly when $d(z_t,z_{\tau_{t-1}})>\eta(S_t)$, after observing $z_t$. Write $V=S_{T-1}=V_T$. This single policy, without knowing $V$ in advance, satisfies
\begin{equation}\label{eq:clipped-cost}
\begin{gathered}
 J_\lambda\le\lambda T,\\
 \begin{aligned}
 J_\lambda\le3\min\bigl\{&\lambda T,\ \lambda+(T-1)V,\\
 &\lambda+2\sqrt{\lambda TV}\bigr\}.
 \end{aligned}
\end{gathered}
\end{equation}
More precisely,
\begin{equation}\label{eq:display-16}
 J_\lambda\le
 \begin{cases}
 \lambda+(T-1)V,&0\le V\le\lambda/T,\\
 (3-1/T)\sqrt{\lambda TV},&\lambda/T<V\le\lambda T,\\
 \lambda T,&V>\lambda T.
 \end{cases}
\end{equation}
The middle inequality is strict.
\end{theorem}
\begin{proof}
The function $\eta$ is positive, nondecreasing, and at most $\lambda$. At a refresh the current distortion is zero and the charge is $\lambda$; otherwise the distortion is at most $\eta(S_t)\le\lambda$ and there is no charge. Each opportunity therefore costs at most $\lambda$, including initialization, proving $J_\lambda\le\lambda T$. If $r_0=0<r_1<\cdots<r_m$ are refresh times, the triangle inequality and strict trigger imply, for $j\ge1$,
\begin{equation}\label{eq:display-17}
 \Delta_j:=S_{r_j}-S_{r_{j-1}}
 \ge d(z_{r_j},z_{r_{j-1}})>\eta(S_{r_j}).
\end{equation}
Hence
\begin{equation}\label{eq:display-18}
 1<\frac{\Delta_j}{\eta(S_{r_j})}
 \le\int_{S_{r_{j-1}}}^{S_{r_j}}\frac{dv}{\eta(v)}.
\end{equation}
The intervals are disjoint. If $m\ge1$, summing gives $m<I(V)$, where direct integration yields
\begin{equation}\label{eq:display-19}
 I(V)=\int_0^V\frac{dv}{\eta(v)}=
 \begin{cases}
 TV/\lambda,&V\le\lambda/T,\\
 2\sqrt{TV/\lambda}-1,&\lambda/T\le V\le\lambda T,\\
 V/\lambda+T-1,&V\ge\lambda T.
 \end{cases}
\end{equation}
If $V\le\lambda/T$, the initial-state discrepancy never exceeds the threshold, so $N=1$ and $J_\lambda\le\lambda+(T-1)V$. In the middle regime, $N<2\sqrt{TV/\lambda}$; this also holds when $m=0$. Distortion is at most $(T-1)\eta(V)$ because its initial term is zero. This proves the strict middle bound. The global cap $J_\lambda\le\lambda T$ gives the last regime.

For comparison with the minimum, in the first regime its smallest entry is $\lambda+(T-1)V$; in the last it is $\lambda T$. In the middle, put $s=\sqrt{\lambda TV}$. Then $\lambda T\ge s$, $\lambda+2s\ge s$, and
\begin{equation}\label{eq:display-20}
 \lambda+(T-1)V\ge2\sqrt{\lambda(T-1)V}\ge s
\end{equation}
because $T\ge2$. Each entry is therefore at least $s$, establishing the factor-three statement.
\end{proof}
The clipped policy costs no more than always refreshing. The factor-three estimate in \eqref{eq:clipped-cost} uses the minimum of three analytical upper bounds. Thresholds based on observed increments also appear in event-triggered reporting \cite{adaptivesod}.

\subsection{The distortion--refresh trade-off}
\begin{theorem}[Finite ramp calculation]\label{thm:ramp}
Let $z_t=vt$ on the real line with $v>0$. Among schedules with exactly $N\in\{1,\ldots,T\}$ refreshes, define $q=\lfloor T/N\rfloor$ and $r=T-Nq$. The exact minimum for $N$ refreshes and a lower bound for the unrestricted optimum are
\begin{equation}\label{eq:ramp-oracle}
\begin{aligned}
 J^*_{\lambda,N}={}&\lambda N+\frac v2\bigl[rq(q+1)\\
 &\qquad +(N-r)q(q-1)\bigr],\\
 J^*_{\lambda}\ge{}&T\sqrt{2\lambda v}-vT/2.
\end{aligned}
\end{equation}
Here $J^*_{\lambda}=\min_{1\le N\le T}J^*_{\lambda,N}$.
\end{theorem}
\begin{proof}
The refreshes partition the horizon into positive integer lengths $\ell_1,\ldots,\ell_N$ totaling $T$. A segment of length $\ell$ has distortion $v\ell(\ell-1)/2$. If two lengths differ by at least two, moving one point from the longer to the shorter decreases their total distortion. Optimal lengths therefore differ by at most one: $r$ lengths are $q+1$, the rest $q$. This gives the exact expression. Finally,
\begin{equation}\label{eq:ramp-proof}
\begin{aligned}
 J_\lambda&=\lambda N+\tfrac v2\Bigl(\sum_j\ell_j^2-T\Bigr)\\
 &\ge\lambda N+\frac{vT^2}{2N}-\frac{vT}{2}\\
 &\ge T\sqrt{2\lambda v}-\frac{vT}{2}.
\end{aligned}
\end{equation}
\end{proof}

\begin{corollary}[Minimax square-root order]\label{cor:minimax}
Consider real-valued paths with $T\ge2$ and variation at most $V>0$. Put $v=V/(T-1)$. If $v\le\lambda\le vT^2$, then even allowing randomized causal policies, with known $V,T,\lambda$,
\begin{equation}\label{eq:minimax-order}
\begin{aligned}
 & (\sqrt2-\tfrac12)\sqrt{\lambda VT}\\
 &\quad\le\inf_\pi\sup_{z:V_T\le V}\mathbb E[J_\lambda(\pi,z)]\\
 &\quad\le(2+\sqrt2)\sqrt{\lambda VT}.
\end{aligned}
\end{equation}
\end{corollary}
\begin{proof}
The threshold upper bound is $\lambda+2\sqrt{\lambda TV}$. In this regime $\lambda/\sqrt{\lambda TV}\le\sqrt{T/(T-1)}\le\sqrt2$, yielding the stated upper constant. For the lower bound use the ramp with variation exactly $V$. Since $\lambda\ge v$, Theorem~\ref{thm:ramp} is at least $(\sqrt2-1/2)T\sqrt{\lambda v}\ge(\sqrt2-1/2)\sqrt{\lambda VT}$. The bound holds for every offline partition, and hence for every outcome of a randomized causal policy.
\end{proof}
The ramp bound in \eqref{eq:ramp-oracle}, and hence the lower side of \eqref{eq:minimax-order}, holds even for offline schedules. It describes the distortion--refresh trade-off shared by all policies. The threshold rule matches this worst-case order in the stated regime.

\paragraph{A burst followed by a plateau}
Consider $T=64$, $\lambda=1$, and a four-step burst followed by a plateau: $z_t=V\min\{t,4\}/4$. For $V=1/2$, the anytime policy refreshes at $0,3,33$. Once the burst ends, it retains $3/8$ until its shrinking threshold falls strictly below the residual $1/8$. Equality at $t=32$ causes no refresh; $t=33$ does. The clipped policy refreshes at $0,1,2,3,4$ and incurs no distortion. Their costs are $7$ and $5$, respectively, against the exact offline cost $11/4$. At $V=2$, the anytime cost is $4$ and the clipped cost is $5$; the offline cost is $4$. Figure~\ref{fig:analytic} separates distortion from refresh charges; the two paths reverse the policies' cost ordering.

\begin{figure*}[t]
\centering
\includegraphics[width=\textwidth]{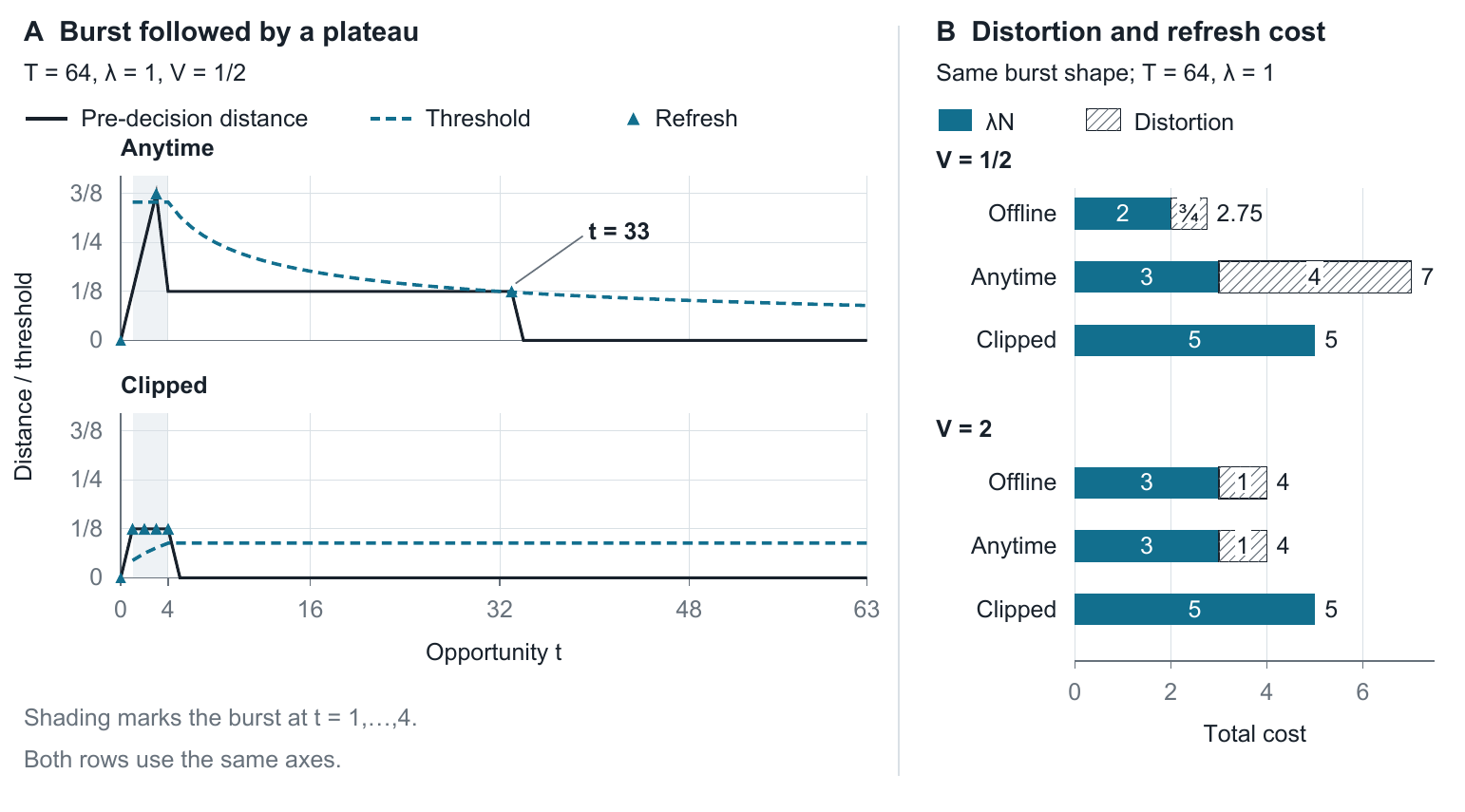}
\caption{Computed costs and update decisions on burst--plateau paths $z_t=V\min\{t,4\}/4$, with $T=64$ and $\lambda=1$. (A) Before-decision distance and threshold for the anytime and clipped rules at $V=1/2$. Markers denote refreshes, after which discrepancy is zero. The strict anytime rule preserves equality at $t=32$ and refreshes at $t=33$. (B) Refresh charge $\lambda N$ and cumulative post-decision distortion at $V=1/2$ and $V=2$. Offline costs use \eqref{eq:offline-oracle}. All costs include initialization.}
\label{fig:analytic}
\end{figure*}

\paragraph{Seeded path families}
We also evaluate 72 distinct paths from random walks, burst--plateau processes, and switching regimes. Each family uses $T\in\{32,64,128\}$ and seeds $0,\ldots,7$, with charges $\lambda\in\{1/16,1/4,1,4\}$, giving 288 path--charge cases and 1,728 cost evaluations, including the offline oracle. The code specifies the generators, seeds, and charges and reports every case. Table~\ref{tab:benchmark} compares the fixed threshold, anytime and clipped rules, always refreshing, and retaining only the initial state against the exact offline oracle. All anytime and clipped costs satisfy their analytical bounds. Clipped costs are lower than anytime in 23 cases, equal in 113, and higher in 152. Per-case records separate distortion and refresh charges and also report $J_\lambda/\sqrt{\lambda TV_T}$.

\begin{table*}[t]
\caption{Synthetic-path cost ratios $J_\lambda/J_\lambda^*$. Each entry is the median (maximum) over 96 path--charge cases per family. The fixed threshold uses the final variation. Costs are computed exactly and ratios are rounded to three decimals.}
\label{tab:benchmark}
\centering\footnotesize
\begin{tabular}{lrrr}
\toprule
Policy & Random walk & Burst--plateau & Regime switches \\
\midrule
Fixed threshold (known $V_T$) & 1.147 (1.625) & 1.071 (3.111) & 1.015 (1.727) \\
Anytime & 1.163 (1.750) & 1.000 (2.684) & 1.061 (2.045) \\
Clipped & 1.585 (2.447) & 1.071 (3.921) & 1.537 (3.184) \\
Always refresh & 4.246 (19.692) & 16.000 (61.134) & 4.453 (23.953) \\
Initial state only & 4.053 (27.313) & 36.875 (1590.250) & 7.013 (52.711) \\
\bottomrule
\end{tabular}
\end{table*}

\subsection{Fixed thresholds and competitive ratios}
A fixed threshold can satisfy a variation bound while having an unbounded refresh-count ratio.
\begin{proposition}[An unbounded update-count ratio]\label{prop:counter}
Under the hard metric-radius constraint and without an effective age cap, the strict threshold policy need not be constant-competitive in refresh count against the offline oracle.
\end{proposition}
\begin{proof}
Choose $\eta>0$ and $0<\epsilon<\eta/2$. Let
\begin{equation}\label{eq:display-21}
 z_0=0,\quad z_1=\eta/2,\quad z_2=\eta+\epsilon,\quad
 z_3=0,\quad z_4=\eta+\epsilon,\ldots
\end{equation}
The threshold policy does not refresh at $1$ and then refreshes at every $t\ge2$, so it uses $T-1$ refreshes. The offline schedule refreshes at $0$ and $1$. Its retained midpoint is within $\eta$ of both later values. One refresh is infeasible because of $z_2$, so the optimum is exactly two. The ratio $(T-1)/2$ diverges.
\end{proof}
The variation of this path grows with its length, so the accounting bound in Theorem~\ref{thm:variation} continues to hold as the refresh-count ratio diverges. Theorem~\ref{thm:hardimpossible} gives the stronger separation for every feasible causal policy. A fixed age cap $K$ changes the comparison: every offline segment has at most $K$ opportunities, giving a refresh-count ratio bound of $K$.

\section{From state drift to output error}\label{sec:fidelity}
We choose the tracking metric so that state discrepancy controls the output change at a fixed current input. Output and value norms are fixed on their respective normed spaces. Assume a fixed known $L>0$ such that
\begin{equation}\label{eq:uniform-sensitivity}
\|F_t(x_t,A(z))-F_t(x_t,A(z'))\|\le L\,d(z,z')
\end{equation}
at every opportunity for all admissible state pairs and current inputs. Let
\begin{equation}\label{eq:display-22}
 E_{\mathrm{out}}=\sum_{t=0}^{T-1}\|y_t^{\rm ca}-y_t^{\rm fr}\|.
\end{equation}
For $\lambda>0$, use Theorem~\ref{thm:anytime} with tracking charge $\lambda/L$, equivalently threshold $\sqrt{\lambda S_t/(Lt)}$. With no age constraint,
\begin{equation}\label{eq:outputbridge}
 E_{\mathrm{out}}+\lambda N
 \le L J_{\lambda/L}
 \le \lambda+4\sqrt{\lambda LTV_T}.
\end{equation}
Choose $L$ for the stated input domain and configure the policy with it before execution.

When the operator space $\mathcal A$ is normed, a sufficient condition is the factorization
\begin{equation}\label{eq:display-23}
\begin{gathered}
 \|F_t(x_t,a)-F_t(x_t,b)\|\le L_{F,t}\|a-b\|,\\
 \|A(z)-A(z')\|\le L_A d(z,z'),
\end{gathered}
\end{equation}
Here $L_A$ and $L_{F,t}$ are the constructor and readout Lipschitz constants; a known uniform bound on $L_{F,t}L_A$ supplies $L$. A scalar loss that is $L_{\ell,t}$-Lipschitz in the output norm adds the factor $L_{\ell,t}$ to the loss-difference bound.

Figure~\ref{fig:certificate} traces the distance-softmax bounds and shows how extra zero coordinates dilute RMS without changing the output.

\begin{figure*}[t]
\centering
\includegraphics[width=\textwidth]{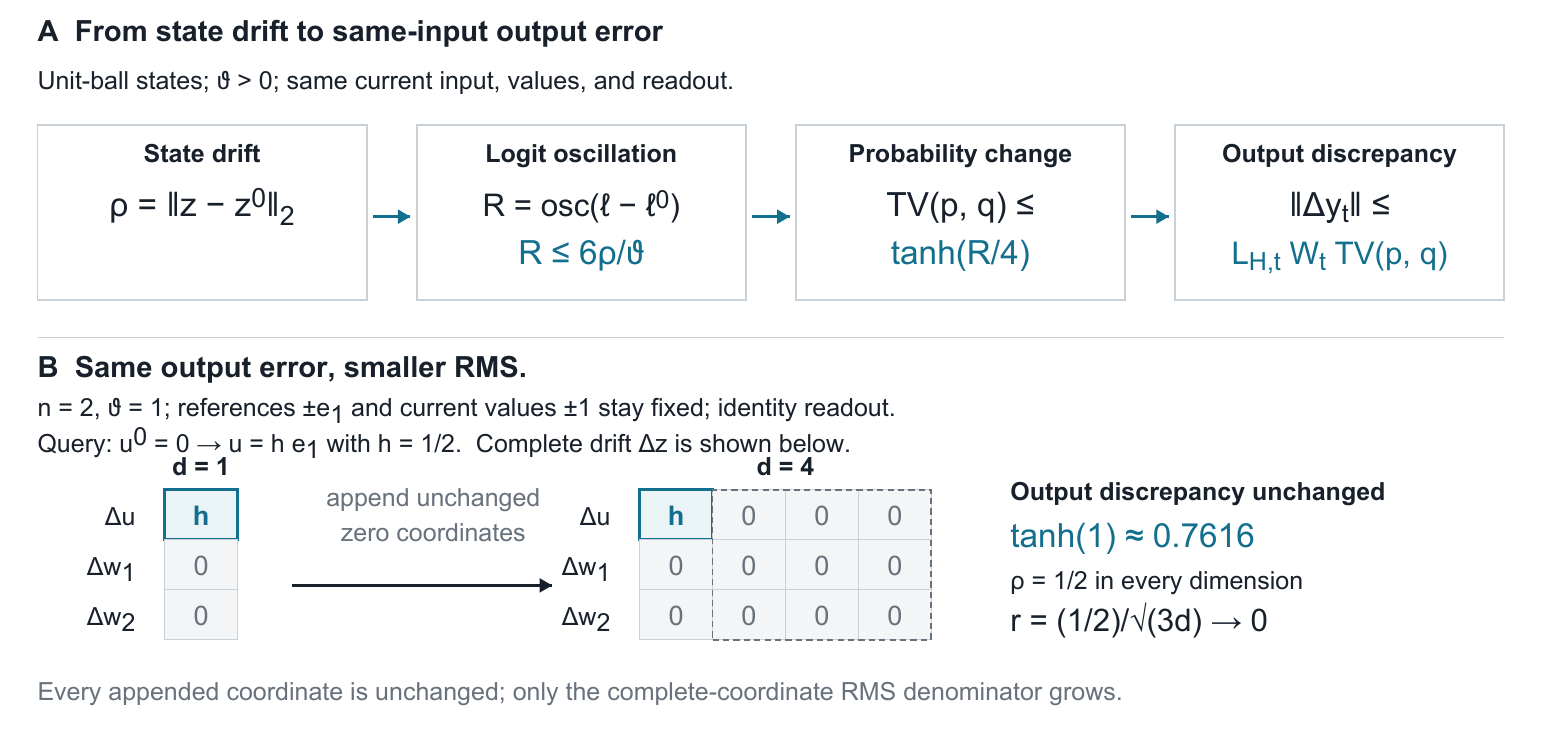}
\caption{Same-input output bounds and RMS dilution. (A) On the product of unit balls, $R\le6\rho/\vartheta$, then $\TV(p,q)\le\tanh(R/4)$ and output discrepancy at most $L_{H,t}W_t\TV(p,q)$. If every query and reference at both endpoints has norm one, $2\sqrt6$ replaces $6$. Each arrow is an upper bound. (B) Proposition~\ref{prop:rms} with $n=2$, $h=1/2$, and $\vartheta=1$: adding unused zero coordinates leaves output discrepancy $\tanh(1)$ fixed while $r=1/(2\sqrt{3d})$ tends to zero.}
\label{fig:certificate}
\end{figure*}

\subsection{Mixture and softmax bounds}
For an integer $n\ge1$, let $\Delta_n=\{a\in\mathbb R^n:a_i\ge0,\ \sum_i a_i=1\}$ be the probability simplex. Consider mixture weights $a(z)\in\Delta_n$ and current values $v_{t,1},\ldots,v_{t,n}$ in a normed space. Define
\begin{equation}\label{eq:display-24}
\begin{gathered}
 m_t(a)=\sum_i a_i v_{t,i},\qquad W_t=\max_{i,j}\|v_{t,i}-v_{t,j}\|,\\
 \TV(p,q)=\tfrac12\|p-q\|_1.
\end{gathered}
\end{equation}
\begin{lemma}[Current-value diameter]\label{lem:diam}
For $p,q\in\Delta_n$,
\begin{equation}\label{eq:display-25}
 \|m_t(p)-m_t(q)\|\le W_t\TV(p,q).
\end{equation}
The coefficient is sharp in general.
\end{lemma}
\begin{proof}
Write the positive and negative parts of $p-q$ as $\alpha r$ and $\alpha s$, where $\alpha=\TV(p,q)$ and $r,s$ are distributions. For $\alpha>0$, the difference is $\alpha\sum_{i,j}r_i s_j(v_{t,i}-v_{t,j})$, whose norm is at most $\alpha W_t$. The zero case is immediate. Two scalar values $0,W_t$ attain equality.
\end{proof}

For finite logit vectors $\ell,\ell^0\in\mathbb R^n$, define $[\softmax(\ell)]_i=e^{\ell_i}/\sum_j e^{\ell_j}$. Let $p=\softmax(\ell)$ and $q=\softmax(\ell^0)$, and write $R=\osc(\ell-\ell^0)=\max_i(\ell_i-\ell_i^0)-\min_i(\ell_i-\ell_i^0)$. The sharp Hilbert-distance comparison \cite{reeb,cohen} gives the following bound.
\begin{lemma}[Standard projective-metric comparison]\label{lem:tv}
\begin{equation}\label{eq:projective-tv}
\TV(p,q)\le\tanh(R/4).
\end{equation}
For every $R\ge0$, a pair of two-category logit vectors attains equality. Here $R$ is the Hilbert projective distance $\osc(\log p-\log q)$.
\end{lemma}
\begin{proof}
Put $r_i=e^{\ell_i-\ell_i^0}$, $m=\min_i r_i$, $M=\max_i r_i$, and $Z=\sum_i q_i r_i$. In $\mathbb E_q$, $r$ takes value $r_i$ with probability $q_i$; write $x_+=\max\{x,0\}$. For $M>m$, convexity of the positive part gives
\begin{equation}\label{eq:display-26}
\begin{aligned}
 \TV(p,q)&=\frac{\mathbb E_q[(r-Z)_+]}{Z}\\
 &\le\frac{(M-Z)(Z-m)}{Z(M-m)}\\
 &\le\frac{\sqrt M-\sqrt m}{\sqrt M+\sqrt m}.
\end{aligned}
\end{equation}
The last maximum occurs at $Z=\sqrt{mM}$, and the last ratio is $\tanh(\log(M/m)/4)=\tanh(R/4)$. A two-point distribution supported at $m,M$ with this mean attains equality. If $M=m$, both distributions coincide.
\end{proof}

\begin{proposition}[Certificate for a current readout]\label{prop:cert}
If $H_t$ is $L_{H,t}$-Lipschitz and $y_t(a)=H_t(m_t(a))$, then
\begin{equation}\label{eq:readout-certificate}
\|y_t(p)-y_t(q)\|\le L_{H,t}W_t\tanh(R/4).
\end{equation}
A loss difference has the additional factor $L_{\ell,t}$ whenever its Lipschitz assumption holds. For a linear readout, one can instead apply Lemma~\ref{lem:diam} to its transformed current values, often giving a tighter certificate.
\end{proposition}
The diameter is invariant under a common translation of the values, which can make a maximum value norm a loose bound.

\begin{proposition}[No value-free output guarantee]\label{prop:no}
Without a bound on current-value sensitivity, arbitrarily small logit change cannot guarantee any fixed output-error tolerance.
\end{proposition}
\begin{proof}
Use two scalar values $0,M$, identity readout, $q=(1/2,1/2)$ and logits changing from $(0,0)$ to $(0,h)$ for any $h>0$. The discrepancy is $M(e^h/(1+e^h)-1/2)$, which is unbounded as $M\to\infty$. One may make $h$ arbitrarily small before choosing $M$. With truth $0$, all outputs are nonnegative, so the absolute-loss difference is the same discrepancy.
\end{proof}

\subsection{Squared-distance softmax}
As a concrete specialization, let $z=(u,w_1,\ldots,w_n)$ concatenate a query and ordered references in $\R^d$, where $d\ge1$ is the vector dimension. Write $z^0=(u^0,w_1^0,\ldots,w_n^0)$ for the retained summary. Every component vector of both summaries has Euclidean norm at most one. For a positive softmax temperature $\vartheta$, set
\begin{equation}\label{eq:distance-softmax}
\begin{gathered}
 \ell_i(z)=-\|u-w_i\|_2^2/\vartheta,\qquad \vartheta>0,\\
 \rho=\|z-z^0\|_2.
\end{gathered}
\end{equation}
Here $\rho$ measures Euclidean drift of the complete concatenated summary. 
\begin{proposition}[Explicit logit-gap bound]\label{prop:gap}
On the product of unit balls,
\begin{equation}\label{eq:logit-gap}
\osc(\ell(z)-\ell(z^0))\le 6\rho/\vartheta.
\end{equation}
Consequently Proposition~\ref{prop:cert} applies with $\tanh(3\rho/(2\vartheta))$. If every query and reference has norm exactly one, the coefficient $6$ can be replaced by $2\sqrt6$.
\end{proposition}
\begin{proof}
For each pair $i,j$, the query norm cancels in $g_{ij}=\vartheta(\ell_i-\ell_j)=\|u-w_j\|^2-\|u-w_i\|^2$. Its gradient in the three relevant vectors has squared norm
\begin{equation}\label{eq:display-27}
 4\bigl(\|w_i-w_j\|^2+\|u-w_i\|^2+\|u-w_j\|^2\bigr)\le36.
\end{equation}
Indeed, the sum of three pairwise squared distances equals three times the sum of squared norms minus the squared norm of their sum, and is at most nine. The product of unit balls is convex, so integrating the gradient on a line segment gives $|g_{ij}(z)-g_{ij}(z^0)|\le6\rho$. Taking the largest pair gives the oscillation bound.

For exactly unit endpoints, $g_{ij}=2u^\top(w_i-w_j)$. Extend this expression to the product of balls along the chord; its squared gradient norm is at most $4(4+1+1)=24$. Its values at the unit endpoints equal the original gap, giving the second bound.
\end{proof}
On the product of unit balls, the metric trigger $\rho\le\eta$ therefore yields
\begin{equation}\label{eq:display-28}
 \|y_t^{\rm ca}-y_t^{\rm fr}\|
 \le L_{H,t}W_t\tanh\!\left(\frac{3\eta}{2\vartheta}\right).
\end{equation}
Bounding the right-hand side uniformly over the allowed inputs gives a hard output tolerance. A known uniform bound $L_{H,t}W_t\le M$ with $M>0$ and $\tanh\xi\le\xi$ for $\xi\ge0$ yields $L=3M/(2\vartheta)$ in \eqref{eq:uniform-sensitivity}, using the full concatenated Euclidean metric on the product of unit balls. Equation~\eqref{eq:outputbridge} applies with this $L$.

If $r=\rho/\sqrt{(n+1)d}$ is a coordinate RMS, the first certificate contains $\tanh(\tfrac32\sqrt{(n+1)d}\,r/\vartheta)$.

\begin{proposition}[Local optimality of the geometric coefficients]\label{prop:geometricsharp}
For $n=2$ and $d\ge2$, the constants $6$ and $2\sqrt6$ in the respective logit-oscillation bounds cannot be replaced by smaller universal constants. Sharpness is local with respect to concatenated Euclidean drift.
\end{proposition}
Appendix~\ref{app:geometricsharp} proves local optimality.

\begin{proposition}[No dimension-free guarantee from complete-coordinate RMS]\label{prop:rms}
Under the unit-ball assumptions, no output-error modulus $\omega(r)$ with $\omega(r)\to0$ as $r\to0$ can hold uniformly over bank size $n$ and dimension $d$ using only complete-coordinate RMS drift, even with current scalar values in $[-1,1]$ and an identity readout.
\end{proposition}
\begin{proof}
Let $e_j$ denote the $j$th standard basis vector. For even $n$, put half of the bank summaries at $e_1$ and half at $-e_1$, and give them respective current values $+1$ and $-1$. Change the query from $0$ to $he_1$ for fixed $0<h\le1$. The cached mixture is zero, and
\begin{equation}\label{eq:rms-witness}
\begin{gathered}
 \rho=h,\qquad r=\frac{h}{\sqrt{(n+1)d}},\\
 |y^{\rm fr}-y^{\rm ca}|=\tanh(2h/\vartheta).
\end{gathered}
\end{equation}
The output discrepancy is independent of $n,d$, while $r$ tends to zero when $n$ grows or unused zero coordinates are appended. A uniform vanishing modulus is impossible.
\end{proof}
Averaging over the full concatenated summary causes the dilution in \eqref{eq:rms-witness}. For this fixed-reference construction, the unscaled query displacement $\|u-u^0\|_2=h$ and full drift $\rho=h$ do not shrink under coordinate padding.

\subsection{Computational cost}\label{sec:cost}
In a sequential cost model, let $B$ be current-input and readout work shared by both implementations, $c$ the fixed construction cost, $h$ the per-opportunity controller cost, and $s$ additional state and initialization work. With no overlap between these costs,
\begin{equation}\label{eq:work-cost}
\begin{gathered}
 C_{\rm fr}=B+Tc,\qquad C_{\rm ca}=B+Nc+Th+s,\\
 C_{\rm ca}<C_{\rm fr}\quad\Longleftrightarrow\quad(T-N)c>Th+s.
\end{gathered}
\end{equation}
Shared summary work enters $B$; work done only for reuse enters $h$. Computing $R$ and $W_t$ also contributes to $h$: exact $R$ may require fresh logits, and the value diameter can be expensive. Latency and energy additionally depend on overlap, transfers, batching, and memory access.

\section{Related work}\label{sec:related}
Online function tracking combines update charges with accumulated penalties. Bienkowski and Schmid allow arbitrary integer representatives and prove randomized lower bounds against oblivious adversaries for linear penalties and hard tolerances \cite{bienkowskischmid}. Yi and Zhang study multidimensional hard-tolerance tracking with arbitrary feasible representatives \cite{yizhang}. Chen and Russell extend metric tracking to movement-dependent costs, including a fixed charge for any nonzero move \cite{chenrussell}. Their pay-if-move guarantee is logarithmic in the largest tolerance-ball cardinality, where $B_\eta(x)=\{z:d(x,z)\le\eta\}$. Table~\ref{tab:tracking-models} compares the action sets and lower-bound constructions.

\begin{table*}[t]
\caption{Tracking actions, initialization, and lower-bound constructions.}
\label{tab:tracking-models}
\centering\footnotesize
\begin{tabular}{>{\raggedright\arraybackslash}p{.14\textwidth}>{\raggedright\arraybackslash}p{.25\textwidth}>{\raggedright\arraybackslash}p{.25\textwidth}>{\raggedright\arraybackslash}p{.26\textwidth}}
\toprule
Axis & Bienkowski--Schmid \cite{bienkowskischmid} & Chen--Russell \cite{chenrussell} & Current resets \\
\midrule
Online / offline actions & Arbitrary integer representatives for both & Arbitrary feasible metric representatives for both & Retain one state or reset to the current observation for both \\[3pt]
Initialization & Initial $f(0)$ is free; Lemma~6 pays one update to select its hidden center & Fixed initialization terms are separate from the asymptotic comparison & Mandatory initial refresh is charged; the penalized witness uses at most two refreshes \\[3pt]
Horizon / precision & Single-block phase lengths are fixed in advance; integer grid & General metric; Appendix~A uses integer candidate intervals & Hard bound: fixed five-point metric. Penalized: known horizon and a growing, finer grid \\[3pt]
Earlier witness & Lemma~6 may select the final singleton before it is observed & Appendix~A retains a final feasible-intersection point & A common prefix makes the eventual center a legal reset \\[3pt]
Scaling / resets & Scaling gives fixed charge, bounded variation and a fixed block horizon; a shorter candidate interval controls prefix cost in our adaptation & Candidate presentation and repeated endpoints adapt the candidate-intersection mechanism & Theorem~\ref{thm:softimpossible} proves the prefixed transfer with offline cost $<7\lambda/2$ and randomized online cost $\ge b\lambda/16$ in expectation \\[3pt]
Quantitative conclusion & Published linear-penalty ratio $\Omega(\log C/\log\log C)$ as integer-grid charge $C$ grows & Randomized $O(\log b_\eta)$ for fixed nonzero-move charge, $b_\eta=\max_x|B_\eta(x)|$ & Theorem~\ref{thm:hardimpossible}: expected $(T+1)/2$ versus exactly two. On its fixed metric, $b_\eta\le5$ \\
\bottomrule
\end{tabular}
\end{table*}

A candidate prefix adapts Chen and Russell's Appendix~A argument to current resets. With radius $2$, first expose $(1,2,0,3)$. Let $P$ be the retained point at its end and choose $c\in\{1,2\}$ with $\Pr(P=c)\le1/2$. Append $(-1,3,-1,3,\ldots)$ if $c=1$, or $(4,0,4,0,\ldots)$ if $c=2$. Offline can retain $1$ initially or refresh at the prefix occurrence of $2$, using at most two refreshes. If $P\ne c$, a suffix endpoint is retained by the second suffix observation; every later opposite endpoint then forces a refresh. Theorem~\ref{thm:hardimpossible} gives a shorter prefix, the stated horizon constant, exactly two offline refreshes, and a strict radius margin. On its fixed five-point metric, $b_\eta\le5$ makes the arbitrary-representative guarantee constant up to fixed initialization terms, while the current-observation ratio grows with $T$.

For the linear-penalty objective, a fixed charge is a normalization. Write $f_t$ for observations and $g_t$ for retained representatives. Starting from unit discrepancy penalty and update charge $C$, scale coordinates and cost by $s=\lambda/C$:
\begin{equation}\label{eq:display-29}
\begin{aligned}
 s\Bigl(\sum_t|f_t-g_t|+CN\Bigr)
 &=\sum_t|sf_t-sg_t|+\lambda N,\\
 V_{\rm scaled}&=(\lambda/C)V_{\rm original}.
\end{aligned}
\end{equation}
Bienkowski and Schmid's Theorem~6 uses a nested-interval block whose phase lengths depend only on $C$ and the phase index. Its observations are nondecreasing left endpoints in an interval of width below $C$. Adjoining an initial $0$ and applying \eqref{eq:display-29} gives charge $\lambda$, variation $V_{\rm scaled}<\lambda$, and a horizon fixed before execution.

Bienkowski and Schmid's Lemma~6 allows offline to select the final singleton before its first observation. To make this comparator current-reset legal, Theorem~\ref{thm:softimpossible} uses $M$ candidates at charge $M^2$ and exposes them in a common prefix. The prefix distortion is $O(M^2)$, so scaling keeps the total offline cost $O(\lambda)$. The nested phases preserve the randomized lower bound and predetermined lengths.

Send-on-delta reporting relates change thresholds to signal variation and reporting rate \cite{miskowicz}. Diaz-Cacho and colleagues adapt thresholds from observed increments to target a mean transmission rate \cite{adaptivesod}. Event-triggered learning \cite{learning} and lazy aggregated gradients \cite{lag} analyze costly updates within learning and optimization dynamics.

Under Poisson updates and a linear age penalty, Ling and Mi obtain an optimal cache refresh interval proportional to the square root of the refresh charge \cite{lingmi}.

Variation-budget stochastic optimization restricts changes in objective functions and compares decisions made before noisy feedback with a dynamic oracle \cite{bgz}. Smoothed online convex optimization observes the current convex cost before choosing an arbitrary real location and charges movement by distance \cite{soco}. The dynamic-regret model with indicator switching costs charges a fixed amount for changing actions in a bounded convex set. Decisions precede observation of the loss, and the comparator is not charged the same switching penalty \cite{indicator}.

TCP acknowledgment trades service charges against waiting costs \cite{dooly}; recent work provides competitive and expert-regret guarantees \cite{ack}. Waiting costs grow until service, whereas retained-state discrepancy can increase or decrease.

The output certificate combines the sharp Hilbert-distance comparison \cite{reeb,cohen} with the diameter of the current values.

\section{Conclusion}
Current-observation resets make missed states irrecoverable until they reappear. Under a hard discrepancy radius, a fixed five-point metric forces horizon-linear expected refresh counts against an offline optimum of two. For the penalized objective, a candidate prefix makes the Bienkowski--Schmid offline witness legal under current resets. The resulting noncompetitiveness holds for randomized policies on bounded-variation paths, with increasingly fine grids. On a fixed finite metric with positive minimum separation, penalized tracking instead admits a competitive policy. Cumulative variation controls the absolute cost of the anytime and clipped policies.

Sensitivity assumptions convert state drift to same-input output upper bounds. For squared-distance softmax, the resulting geometric coefficients are locally optimal, while the RMS example shows how unused coordinates can conceal a fixed output error.

\paragraph{Code and data}
Implementations, exact checks, figure scripts, and synthetic benchmark records are available at \url{https://github.com/usamihiroyasu/missed-states}.

\appendices
\section{Local optimality of the softmax coefficients}\label{app:geometricsharp}
\begin{proof}[Proof of Proposition~\ref{prop:geometricsharp}]
Write $a=w_1,b=w_2$ and
\begin{equation}\label{eq:display-30}
 G(u,a,b)=\|u-b\|^2-\|u-a\|^2,
 \qquad \ell_1-\ell_2=G/\vartheta.
\end{equation}
For two categories, logit-change oscillation is $|\Delta G|/\vartheta$. In the unit-ball case use an interior radius $0<r<1$ and set
\begin{equation}\label{eq:display-31}
\begin{gathered}
 u_r=r(1,0),\quad a_r=r(-1/2,\sqrt3/2),\\
 b_r=r(-1/2,-\sqrt3/2).
\end{gathered}
\end{equation}
These are interior points and $G=0$. Its full gradient is
\begin{equation}\label{eq:display-32}
\begin{gathered}
 \nabla G=(2(a-b),\,2(u-a),\,2(b-u)),\\
 \|\nabla G(u_r,a_r,b_r)\|=6r.
\end{gathered}
\end{equation}
A sufficiently small perturbation in the unit gradient direction remains feasible and makes $|\Delta G|/\rho$ tend to $6r$. Let $r\uparrow1$.

For exact unit vectors take a real parameter $s$ near zero and set
\begin{equation}\label{eq:display-33}
\begin{gathered}
 u_s=\cos(4s)e_2+\sin(4s)e_1,\\
 a_s=\cos(2s)e_1+\sin(2s)e_2,\qquad b_s=-a_s.
\end{gathered}
\end{equation}
Write $G(s)=G(u_s,a_s,b_s)$ and measure $\rho(s)$ from $(u_0,a_0,b_0)$. All vectors remain unit length and
\begin{equation}\label{eq:display-34}
\begin{aligned}
 G(s)&=4\sin(6s),\\
 \rho(s)^2&=6-2\cos(4s)-4\cos(2s).
\end{aligned}
\end{equation}
Thus $|\Delta G|/\rho\to2\sqrt6$. The two universal upper bounds are therefore optimal. With fixed scalar values $+1,-1$, the mixture equals $\tanh(G/(2\vartheta))$; its derivative at $G=0$ also attains the corresponding composed certificate's first-order coefficient.
\end{proof}

\balance

\begin{IEEEbiography}[{\includegraphics[width=1in,height=1.25in,keepaspectratio,trim=325bp 94bp 33bp 40bp,clip]{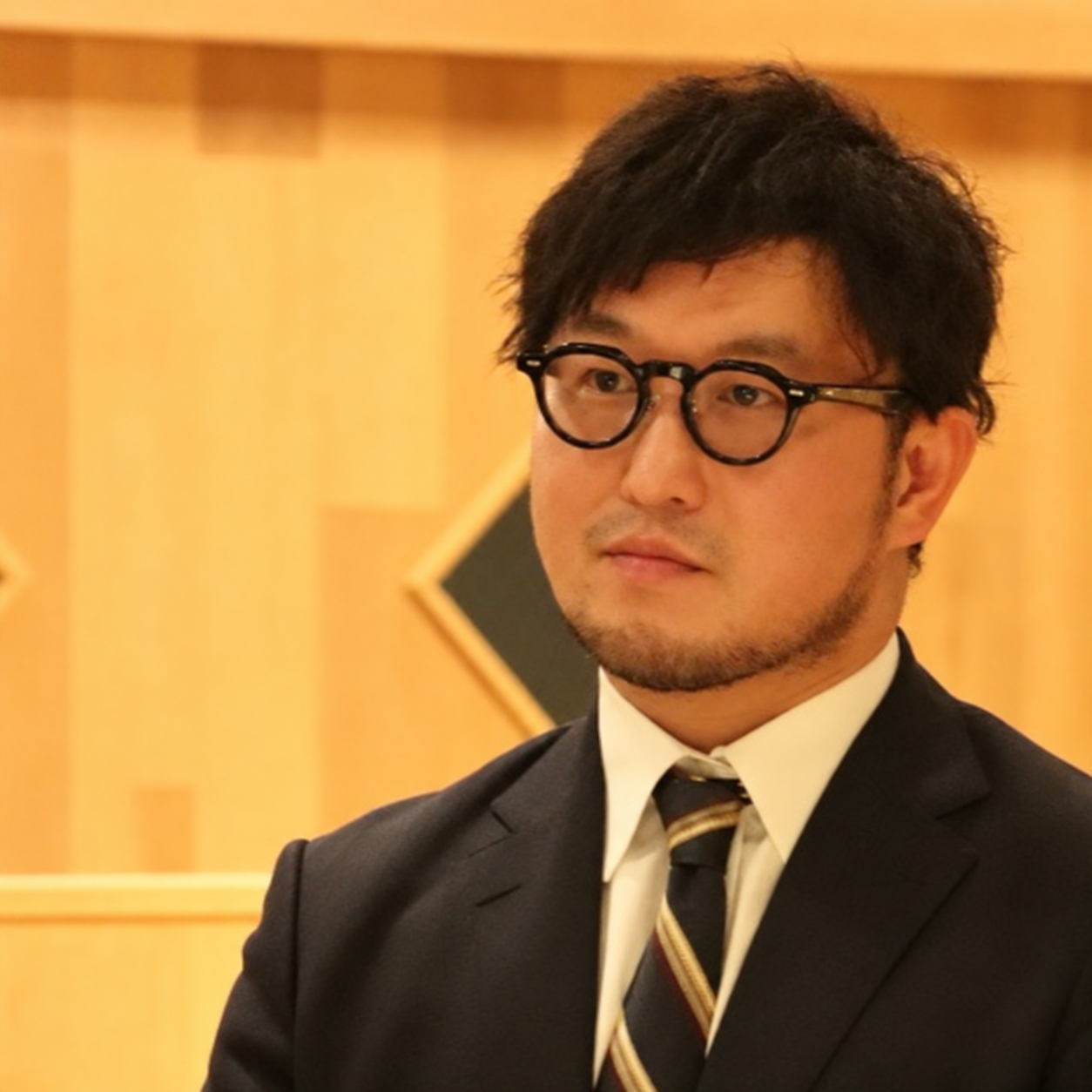}}]{Hiroyasu Usami}
received the Ph.D. degree in engineering from Chubu University in 2019. He is a Senior Assistant Professor at Chubu University and a Researcher at Aichi Medical University. His research focuses on computer vision, machine learning, and medical image analysis. He serves on the Editorial Board of \emph{Scientific Reports} and the editorial committee of the \emph{IPSJ Journal/Journal of Information Processing}.
\EOD
\end{IEEEbiography}

\begin{thebibliography}{99}
\bibitem{bienkowskischmid} M. Bienkowski and S. Schmid, ``Online function tracking with generalized penalties,'' in \emph{Proc. SWAT}, LNCS, vol. 6139, 2010, pp. 359--370. doi: \href{https://doi.org/10.1007/978-3-642-13731-0_34}{\nolinkurl{10.1007/978-3-642-13731-0_34}}.
\bibitem{chenrussell} S. Chen and A. Russell, ``Online metric tracking and smoothing,'' \emph{Algorithmica}, vol. 68, no. 1, pp. 133--151, 2014. doi: \href{https://doi.org/10.1007/s00453-012-9669-8}{\nolinkurl{10.1007/s00453-012-9669-8}}.
\bibitem{yizhang} K. Yi and Q. Zhang, ``Multi-dimensional online tracking,'' \emph{ACM Transactions on Algorithms}, vol. 8, no. 2, pp. 12:1--12:16, 2012. Preliminary version in SODA, 2009. Author manuscript: \url{https://cse.hkust.edu.hk/~yike/onlinetrack-j.pdf}.
\bibitem{reeb} D. Reeb, M. J. Kastoryano, and M. M. Wolf, ``Hilbert's projective metric in quantum information theory,'' \emph{Journal of Mathematical Physics}, vol. 52, Art. no. 082201, 2011. \href{https://arxiv.org/abs/1102.5170}{arXiv:1102.5170}.
\bibitem{cohen} S. N. Cohen and E. Fausti, ``Hyperbolic contractivity and the Hilbert metric on probability measures,'' \href{https://arxiv.org/abs/2309.02413v2}{arXiv:2309.02413v2}, 2024.
\bibitem{miskowicz} M. Mi\'skowicz, ``Send-on-delta concept: an event-based data reporting strategy,'' \emph{Sensors}, vol. 6, pp. 49--63, 2006. doi: \href{https://doi.org/10.3390/s6010049}{\nolinkurl{10.3390/s6010049}}.
\bibitem{adaptivesod} M. Diaz-Cacho, E. Delgado, A. Barreiro, and P. Falc\'on, ``Basic send-on-delta sampling for signal tracking-error reduction,'' \emph{Sensors}, vol. 17, no. 2, Art. no. 312, 2017. doi: \href{https://doi.org/10.3390/s17020312}{\nolinkurl{10.3390/s17020312}}.
\bibitem{learning} F. Solowjow and S. Trimpe, ``Event-triggered learning,'' \emph{Automatica}, vol. 117, Art. no. 109009, 2020. doi: \href{https://doi.org/10.1016/j.automatica.2020.109009}{\nolinkurl{10.1016/j.automatica.2020.109009}}.
\bibitem{lag} T. Chen, G. B. Giannakis, T. Sun, and W. Yin, ``LAG: Lazily aggregated gradient for communication-efficient distributed learning,'' in \emph{Proc. NeurIPS}, 2018. \href{https://arxiv.org/abs/1805.09965}{arXiv:1805.09965}.
\bibitem{lingmi} Y. Ling and J. Mi, ``An optimal trade-off between content freshness and refresh cost,'' \emph{Journal of Applied Probability}, vol. 41, no. 3, pp. 721--734, 2004. doi: \href{https://doi.org/10.1239/jap/1091543421}{\nolinkurl{10.1239/jap/1091543421}}.
\bibitem{bgz} O. Besbes, Y. Gur, and A. Zeevi, ``Non-stationary stochastic optimization,'' \emph{Operations Research}, vol. 63, no. 5, pp. 1227--1244, 2015. \href{https://arxiv.org/abs/1307.5449}{arXiv:1307.5449}.
\bibitem{soco} N. Bansal, A. Gupta, R. Krishnaswamy, K. Pruhs, K. Schewior, and C. Stein, ``A 2-competitive algorithm for online convex optimization with switching costs,'' in \emph{Proc. APPROX/RANDOM}, LIPIcs, vol. 40, 2015, pp. 96--109. doi: \href{https://doi.org/10.4230/LIPIcs.APPROX-RANDOM.2015.96}{\nolinkurl{10.4230/LIPIcs.APPROX-RANDOM.2015.96}}.
\bibitem{indicator} N. Mhaisen and G. Iosifidis, ``Dynamic regret in online convex optimization with indicator switching costs,'' \href{https://arxiv.org/abs/2609.30556}{arXiv:2609.30556v1}, 2026. Preprint.
\bibitem{dooly} D. R. Dooly, S. A. Goldman, and S. D. Scott, ``On-line analysis of the TCP acknowledgment delay problem,'' \emph{Journal of the ACM}, vol. 48, no. 2, pp. 243--273, 2001. doi: \href{https://doi.org/10.1145/375827.375843}{\nolinkurl{10.1145/375827.375843}}.
\bibitem{ack} T. Harada, Y. Kawase, and H. Sumita, ``Simultaneous competitive and regret guarantees for TCP acknowledgment,'' \href{https://arxiv.org/abs/2610.13571}{arXiv:2610.13571v1}, 2026. Preprint.
\end{thebibliography}
\end{document}